\documentclass{article}
\usepackage{amsthm,amsmath,amssymb,amsfonts,braket,hyperref}
\usepackage{graphicx}
\usepackage{xcolor}

\newtheorem{theorem}{Theorem}
\newtheorem{corollary}{Corollary}[theorem]
\newtheorem{lemma}[theorem]{Lemma}
\newtheorem{definition}{Definition}

\begin{document}
\title{Quantum to Classical One-Way Function and Its Applications in Quantum Money Authentication\thanks{Part of this work was done while the first author was visiting R. C. Bose Centre for Cryptology and Security, Indian Statistical Institute, Kolkata during the Summer of 2017 (between the 2nd and the 3rd semester of his BSc (Honours) Mathematics and Computer Science course) for internship under the supervision of the second author.}
}
\author{Amit Behera\\
Department of Mathematics \& Computer Science,\\
Chennai Mathematical Institute, Chennai, India,\\ 
E-mail: \texttt{amitb@cmi.ac.in}
\and
Goutam Paul\\
Cryptology and Security Research Unit,\\
R. C. Bose Centre for Cryptology and Security,\\
India Statistical Institute, Kolkata, India,\\ 
E-mail: \texttt{goutam.paul@isical.ac.in}
}

\date{}

\maketitle

\begin{abstract}
 In 2013, Farid and Vasiliev [arXiv:quant-ph/1310.4922] for the first time proposed a way to construct a protocol for the realisation of ``{\em Classical to Quantum}" one-way hash function, a derivative of the Quantum one-way function as defined by Gottesman and Chuang [Technical Report
arXiv:quant-ph/0105032] and used it for constructing quantum digital signatures.
 We, on the other hand, for the first time, propose the idea of a different kind of one-way function, which is ``{\em quantum-classical}" in nature, that is, it takes an $n$-qubit quantum state of a definite kind as its input and produces a classical output. We formally define such a one-way function  and propose a way to construct and realise it. The proposed one-way function turns out to be very useful in authenticating a quantum state in any quantum money scheme and so we can construct many different quantum money schemes based on such a one-way function. Later in the paper, we also give explicit constructions of some interesting quantum money schemes like quantum bitcoins and quantum currency schemes, solely based on the proposed one-way function. The security of such schemes can be explained on the basis of the security of the underlying one-way functions. 
\end{abstract}


\section{Introduction}
The concept of one-way functions (OWF) is at the core of theoretical foundations of Cryptology. The one-way functions have served as a very useful primitive in many cryptographic protocols. In the post-quantum era, extensive work has been done in quantum one-way functions~\cite{qhash,qdsign} as well. Those one-way functions are ``{\em classical-quantum}" in nature i.e., they take classical bit strings as the input and produce quantum states as the output, satisfying the one-way property of a one-way function. Hence, many information-theoretically secured digital signatures scheme~\cite{qdsign,qps,qdsof} have been devised based on the one-wayness property of these functions. Although, those one-way functions are very effective in authenticating classical bit strings but they fail to authenticate quantum states. For authentication of quantum states, we need a different kind of one-way function  which shall be $quantum$-$classical$ in nature. In our paper, we conceive the idea of such a one-way function  which takes a quantum state (of a particular kind) as its input and produces a classical string, with the help of which quantum states can be authenticated multiple times. 

The biggest difficulty or challenge in realising such a one-way function  is that any quantum transformation of quantum states can be represented by a unitary transformation and hence are linear and invertible. It is not possible to implement a non-linear function by quantum gates. Therefore, with the knowledge of the output state $U\Ket{\psi}$ of any quantum operation $U$, one can trace back to the original input quantum state $\Ket{\psi}$. Hence, any function developed through these operations will become invertible, provided the complete output state is revealed. Instead, if we don't reveal the whole output state $U\Ket{\psi}$, but reveal only partial information about the output state and use this partial information to construct our classical image of the function for the state $\Ket{\psi}$, then the procedure might not be invertible. This is the central idea behind the construction of the proposed one-way function, which we shall formalise later in this paper.

\section{OWF Protocol}
First, we introduce a few relevant notations and definitions and then describe how our OWF can be evaluated and verified.

\subsection{Notations and  Definitions}

Let $\mathcal{C}$ denote the computational basis consisting of two elements $\Ket {0}$ and $\Ket{1}$; and let $\mathcal{H}$ denote the Hadamard basis consisting of two elements $\Ket{+} = \frac{1}{\sqrt{2}}(\Ket{0} + \Ket{1})$ and $\Ket{-} = \frac{1}{\sqrt{2}}(\Ket{0} - \Ket{1}).$

Let ${\mathcal{B}}_n$ $(n>1)$ be the orthonormal basis of $n$-qubit GHZ states whose elements are denoted as
\begin{eqnarray*}
\Ket{y_1,y_2,\ldots,y_n} & = & \frac {1}{\sqrt{2}} (\Ket{y_1}_1 \otimes \Ket{y_2}_2 \otimes \ldots \otimes \Ket{y_n}_n\\
& & + \Ket{{\overline{y}}_1}_1 \otimes \Ket{{\overline{y}}_2}_2 \otimes \ldots \otimes \Ket{{\overline{y}}_n}_n),
\end{eqnarray*}
where $y_i \in \{0,1\} \quad \forall i = 1,2,\ldots,n$ and $\overline{y}$ denotes one's complement of $y$ in binary base. The subscript outside the Ket notation denotes the position of the qubit.

We say, the size of any GHZ state in ${\mathcal{B}}_n$ is $n$.
Let $\mathcal{G}$ be the set of all GHZ states with size less than equal to $n$\footnote{The set $\mathcal{G}$  depends upon $n$ and should be written as $\mathcal{G}(n)$. Since we would be using  $\mathcal{G}(n)$ for a fixed $n$ throughout the paper, we simply refer to the set as $\mathcal{G}$.}.

Next, we will define a special kind of quantum states, (termed as $GCH$ states)  that we will be mostly concerned with. Let $\Ket {g} \in \mathcal{G}$,  $\Ket{c} \in \mathcal{C}$, $\Ket{h} \in \mathcal{H}$. The state $\Ket{g}\otimes \Ket{c}\otimes \Ket{h}$ is called a $GCH$ state. More formally,
\begin{definition}
\label{gch}
An $n$-qubit quantum state, which can be written as a tensor product of either elements of Hadamard basis ($\mathcal{H}$) or of computational basis ($\mathcal{C}$) or of GHZ states ($\mathcal{G}$), up to some rearrangement of the indices of the qubits, is called a $GCH$ state and it is denoted by $\Ket{\psi^{(n)}}_{GCH}$. We say that the \textbf{basis} of a particular qubit of a $GCH$ state is $\mathcal{H}$ or $\mathcal{C}$ or $\mathcal{G}$, if the quantum state of that qubit is an element of  $\mathcal{H}$ or $\mathcal{C}$ or is entangled in a GHZ state ($\mathcal{G}$).
\end{definition}

\begin{definition}
A pair of qubits, say the $k^{th}$ and $l^{th}$ qubits, is said to be \textbf{compatible for CNOT gate operation} such that the $l^{th}$ qubit is used as the control qubit and the $k^{th}$ as the target qubit, iff it is exactly one of the following cases.
         \begin{enumerate}
            \item Both the qubits are in $\mathcal{C}$ or $\mathcal{H}$ or are entangled in a GHZ state (of $\mathcal{G}$).
            \item The $l^{th}$ and the $k^{th}$ qubits are in $\mathcal{H}$ and $\mathcal{C}$ or $\mathcal{C}$ and $\mathcal{H}$ or $\mathcal{G}$ and $\mathcal{C}$ or $\mathcal{C}$ and $\mathcal{G}$ or $\mathcal{G}$ and $\mathcal{H}$ respectively.
         \end{enumerate}
\end{definition}

\begin{definition}
A parallel combination of CNOT gates on compatible pairs of qubits of a $GCH$ state is known as \textbf{OWF gate operation}.
\end{definition}

\begin{definition}
A sequence of \textbf{OWF gate operations} applied in series, is called a \textbf{OWF unitary operation}.
\end{definition}

For an $n$-qubit state, the quantum state of the qubits at each position or index keeps on changing with the gate operations but the positions remain fixed. Now we will define some properties for these positions, that would be useful in describing our protocol.

\begin{definition}
After a \textbf{OWF unitary operation} on an $n$-qubit quantum $GCH$ state, the $k^{th}$ position is said to be \textbf{saturated} (otherwise, it is called \textbf{unsaturated}) if the qubit in the $k^{th}$ position has been used (either as a target or a control qubit) in at least one of the CNOT operations of the \textbf{OWF unitary operation}.  
\end{definition}

Now we will define some of the properties of pairs of positions in a $GCH$ state with respect to \textbf{OWF unitary operation}. Note that here pairs mean unordered pairs.

\begin{definition}
After a \textbf{OWF unitary operation} on an $n$-qubit quantum $GCH$ state, a pair of positions $(k,l)$ in that $GCH$ state is called \textbf{CNOT-once}, if there was at most one CNOT operation on the qubits at this pair of positions either with the $k^{th}$ qubit as the control qubit and the $l^{th}$ qubit as the target qubit or vice versa.
\end{definition}

\begin{definition}
In a \textbf{OWF unitary operation} on a $GCH$ state, a \textbf{CNOT-flipped operation} (abbreviated as \textbf{CF}) on a pair of \textbf{CNOT-once} positions $(k,l)$ in a $GCH$ state is a CNOT operation on the qubits at these positions such that, if there was no CNOT operation with the $k^{th}$ and the $l^{th}$ qubits as inputs, then anyone of them can be the target qubits, and if previously there was a CNOT operation with the $k^{th}$ and the $l^{th}$ qubit as the target and control qubits, then the target and control qubits would be the $l^{th}$ and the $k^{th}$ qubit respectively.
\end{definition}
Note that, we have only considered \textbf{CNOT-once} positions and hence previously there could have been at most one CNOT operation on the qubits at those positions.

Next, we define a special kind of \textbf{OWF gate operation} that will terminate a \textbf{OWF unitary operation}. The OWF termination involves $\frac{n}{2}$ CNOT operations, and is applicable only on even-sized $GCH$ state having only even-sized GHZ states along with some additional properties. Hence, if there are odd number of qubits in $\mathcal{C}$, then the number of qubits in $\mathcal{H}$ is also odd.

\begin{definition}
The \textbf{OWF termination} is a special kind of \textbf{OWF gate operation}, acting on an $n$-qubit $GCH$ state ($n$ is even) having only even-sized GHZ states (and some other properties, stated below), that terminates a \textbf{OWF unitary operation}, after which some of the qubits are chosen as marked qubits. The CNOT operations on compatible pairs of positions are chosen in the following way.
\begin{enumerate}
\item  For each of the GHZ state (having even size), which is an element of $\mathcal{B}_m$, group the positions of the qubits in the GHZ state in pairs and apply a \textbf{CF} operation on each of the pairs of positions. If $m = 2$, mark the position (index) of the qubits that were used as the control qubits in these CNOT gate operations; else mark the positions of the qubits that were used as the target qubits.
\item If there are odd number of qubits in $\mathcal{C}$, choose two qubits, one from $\mathcal{C}$ and other one from $\mathcal{H}$, and apply a \textbf{CF} operation at the positions of these qubits. After this step, each of the remaining number of qubits in $\mathcal{C}$ or $\mathcal{H}$ would be even. 
\item For the remaining qubits, group the positions of the qubits in $\mathcal{C}$ in pairs, and on each pair of positions apply a \textbf{CF} operation. Mark the position (index) of the qubits used as the target qubit in these CNOT gates. Do the same for the qubits in $\mathcal{H}$, with a slight change that instead of marking the position of the target qubits mark the position of the control qubits.
\end{enumerate}
\end{definition}

The $GCH$ state must satisfy that any such pairs of positions chosen as stated in points 1. 2. and 3. of the above definition (for \textbf{CF} operations) are \textbf{CNOT-once}. Otherwise, the \textbf{CF} operations do not make sense. This is an important property that the $GCH$ state must satisfy in order to undergo \textbf{OWF termination}.

\begin{definition}
 The \textbf{quantum to classical one-way function} $\mathcal{F}$ is an algorithm that maps an $n$-qubit quantum $GCH$ state (for any $n$) to an $l$-bit classical string.

 Hence, 
$$ \mathcal{F} : \left\{\Ket {\psi^{(n)}}_{GCH}\right\} \rightarrow \{0, 1\}^{l},$$
where $\Ket {\psi^{(n)}}_{GCH}$ is as defined in Definition~\ref{gch}.
\end{definition}


Now we will define negligible functions and success probability of an adversary. By an adversary, we mean an algorithm to execute forgery against a protocol and by an efficient adversary, we refer to those algorithms that run in polynomial time. Henceforth, all adversaries referred in this paper will be considered as efficient adversaries.

\begin{definition}
A function $f(x)$ is said to be \textbf{negligible}, if for every polynomial function $P(x)$, $$\exists N\in\mathbb{N}\text{ such that } \forall n > N,
f(n) < \frac{1}{P(n)}.$$
\end{definition}

\begin{definition}
For the OWF protocol with public parameters $P$ (defined in the OWF evaluation section), for any arbitrary input state $\Ket{\psi}$ in the domain $S$, and an efficient adversary $A$, we define \textbf{the success probability of forgery } for $A$ as $$Succ(A) = \Pr\left(\mathcal{F}\Ket{\psi}=\mathcal{F}\Ket{\phi}~\left.\right|~ A(P) = \Ket{\phi}\right),$$ 
\end{definition}
where $A(P)$ denotes the output of the adversary with public parameters as inputs.

In various cryptographic models and schemes, the classical one-way function is used for data authentication because of their one-wayness. In this paper also, we will be chiefly concerned about the use of the proposed quantum one-way function for data authentication. At this point, we should note a significant and fundamental difference between quantum and classical data. In the classical  setting, given classical data, we can always figure out the classical state of the data; but given an unknown quantum data, it is not possible to determine its actual quantum state, with $100\%$ certainty. 

In the authentication of classical data, in order to verify whether a given data $A$ is the same as $B$ (whose image $\mathcal{F}(B)$ under a given classical one-way function $\mathcal{F}$ is evaluated and known publicly), we evaluate $\mathcal{F}(A)$ and compare the two. Hence, in classical setting, verification through one-way function involves separate evaluation of the one-way function on the data bit strings $A$ and $B$ followed by comparison.
On the other hand, in the quantum setting, we need not know the actual quantum state of data $A$ and hence cannot repeat the evaluation algorithm on $A$ in the verification process.

For this reason, for the proposed quantum one-way function, we have included two algorithms, one for function evaluation and the other for verification.

\subsection{\textbf{OWF Evaluation}}
The OWF evaluation algorithm may be summed up in the following steps.
\begin{enumerate}
\item For an input $n$-qubit $GCH$ state $\Ket{\psi}$, a \textbf{OWF unitary operation} is constructed with $k$ many \textbf{OWF gate operations}, $X_1,X_2,\ldots,X_k$, in the following manner. 

 $$ \text{Let } \Ket {\phi_j} = X_j(X_{j-1}\ldots(X_2(X_1\Ket {\psi}))\ldots)$$ 
\hfill $\qquad \qquad \forall j = 1,2,\ldots k.$

\begin{enumerate}
 \item The \textbf{OWF gate operations} are chosen such that after the first $(k-2)$ gate operations, all possible pairs of positions in the $GCH$ state are \textbf{CNOT-once} [Theorem \ref{thm_1} of section III A ensures that such choices are possible].
 \item The value of $k$  is chosen sufficiently large so that each of the $n$ positions has been saturated after the first $(k-2)$ \textbf{OWF gate operations} [Lemma \ref{lemma_2} of section III A shows that we can do so without violating the previous property].
 \item For the $(k-1)^{th}$ gate, the CNOT gates are chosen in the following way.

For each of the GHZ state with odd number of qubits, choose any two qubits in the GHZ state and apply a \textbf{CF} operation on their positions. This makes sense as these pair of positions were \textbf{CNOT-once}. After this operation, there will be even number of qubits in each of the GHZ states.
 \item The final gate or the $k^{th}$ gate is a OWF termination operation, involving $\frac{n}{2}$ CNOT operations after which at least $(\frac{n}{2} - 1)$ qubits are marked. Note that the pairs of positions chosen during OWF termination are different from the the pairs of positions used in \textbf{CF} operations used in the $(k-1)^{th}$ gate. Hence the  pairs of positions chosen during OWF termination were \textbf{CNOT-once}. From the set of marked qubits, $(\frac{n}{2} - 1)$ qubits are randomly chosen, say, the ${m_1}^{th},{m_2}^{th},\ldots,{m_{(\frac{n}{2} - 1)}}^{th}$ qubits. The total combined quantum state ${\Ket {\alpha}}_{m_1,m_2,\ldots,m_{(\frac{n}{2} - 1)}}$ is noted down as a classical information $C_l$ and the position (in our case $m_1,m_2,\ldots,m_{(\frac{n}{2} - 1)}$) of the chosen marked qubits are noted down as the \textbf{\textit{measurement position}} of the \textbf{OWF unitary operation} $U_l$. The projective POVM's~\cite{qNC02} for each index $i$ of the measurement position (for measuring the $i^{th}$ qubit of that index non-destructively), be termed as \textbf{\textit{measurement basis}}. Hence, for all but one CNOT operations of $X_k$, one of its output qubit is marked and its quantum state is recorded in $C_l$. Let the description of the sequence of gates and the  measurement positions and measurement basis  in each of the \textbf{OWF unitary operations} $U_l$ be denoted by $F_l$ and we define $F_l\Ket{\psi} = C_l$.  
\end{enumerate}

\item The first two steps (1) and (2) are iterated $n$-times with different sets of sequential quantum gates, where the $l^{th}$ iteration corresponds to the \textbf{OWF unitary operation} $U_l$. Each of the $U_l$ operations are independent of one another. Let $C = (C_1\|C_2\ldots \|C_n).$
\item Then the one-way function ($\mathcal{F}$) can be described as $$\mathcal{F} = (F_1\|F_2|\cdots\|F_n)$$ and we define the image $\mathcal{F}\Ket{\phi}$ as the concatenation of the classical measured value of the combined quantum state in each of the $U_l$'s. Hence, in the above algorithm for the input state $\Ket{\psi}$, $$\mathcal{F}\Ket{\psi} = (C_1\|C_2\ldots \|C_n) = C .$$
The one-way function description $\mathcal{F}$ and the image $C$ constitute the public parameters of the OWF protocol for $\Ket{\psi}$.
\end{enumerate}
In the diagram below, we illustrate the action of the OWF evaluation procedure on a $4$-qubit $GCH$ state $(\Ket{00}\otimes\Ket{1}\otimes\Ket{+})$. The $4$ \textbf{OWF unitary operations} $\text{U}_1,\text{U}_2,\text{U}_3,\text{U}_4$ yield the $4$ chosen quantum states, the classical measured value of which $\text{C}_1,\text{C}_2,\text{C}_3,\text{C}_4$ together construct the image of the one-way function.
\begin{figure}[!ht]
\includegraphics[width=\textwidth]{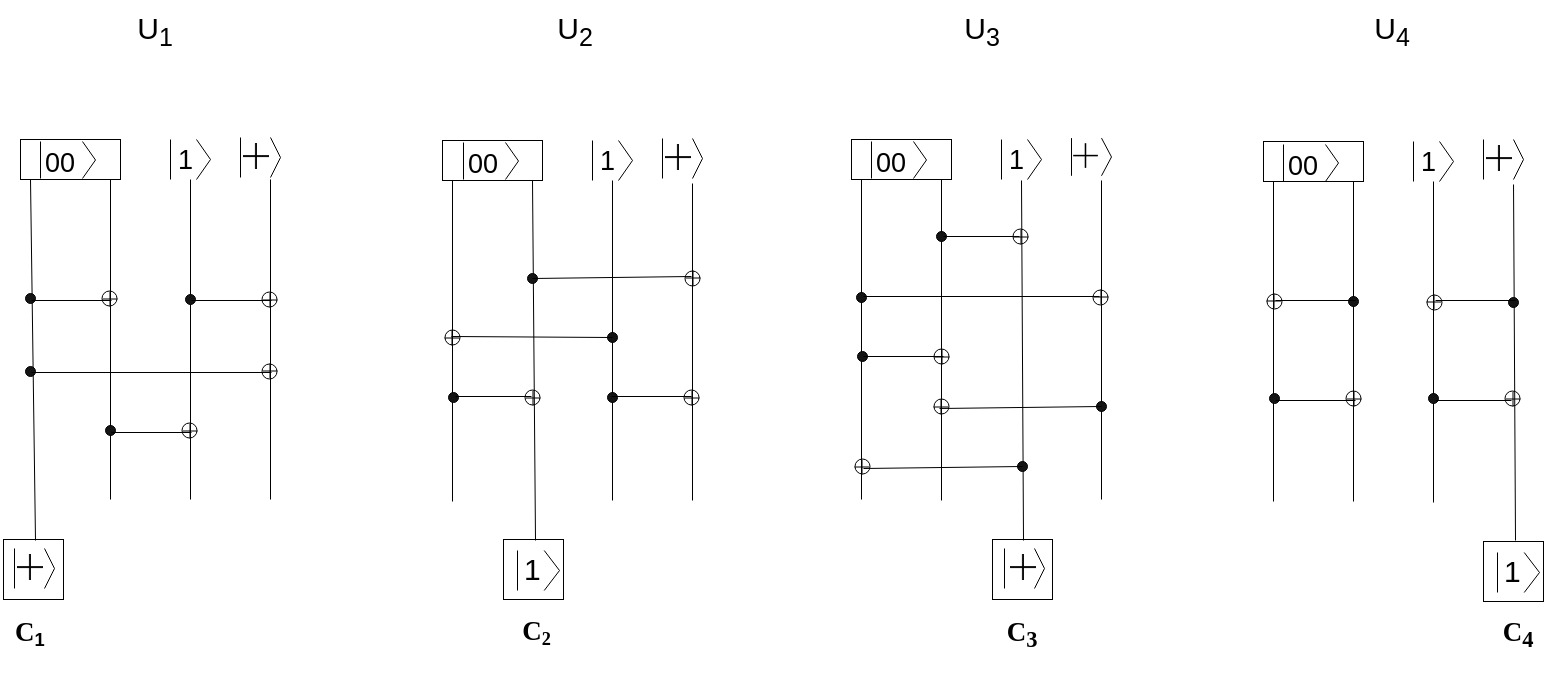}
\caption{Instance of OWF evaluation algorithm}
\label{fig1}
\end{figure}
In the OWF evaluation algorithm, for each \textbf{OWF unitary operation}, after the \textbf{OWF termination} involving exactly $\frac{n}{2}$ gate operations, depending upon the initial choice of the first  $(k-2)$ OWF gate operations, there can be two possible cases: either $\frac{n}{2}$ qubits are marked or $(\frac{n}{2} - 1)$ qubits are marked. Among the marked qubits, exactly $(\frac{n}{2} - 1)$ are randomly chosen such that their positions are declared to be the measurement positions. This is crucial in order to make sure that the adversary who does not know the input state cannot distinguish between the two cases. Hence, for each \textbf{OWF unitary operation}, there will be exactly $1$ out of $\frac{n}{2}$ CNOT operations in the \textbf{OWF termination}, none of whose output qubits will be a part of the qubits in the measurement positions. In Fig.~\ref{fig1}, for each $U_l$, there is exactly one such unmeasured CNOT gate. Although, the use of such unmeasured CNOT gates in OWF evaluation is vital for security aspects, it is also true that these gates are not required in OWF verification and hence their details can be removed from the one-way function description.

\subsection{\textbf{OWF Image Verification}}
After receiving a quantum state $\Ket{\phi}$ and the description of the one-way function $\mathcal{F}$ (i.e., the description of the $n$ quantum operations $U_1,U_2,\ldots,U_n,$ including their measurement positions and measurement basis) and the one-way function value $C$, an user runs the following verification algorithm.
 \begin{itemize}
 \item At first, for each of the \textbf{OWF unitary operations} $U_l$, the user applies sequentially the quantum gates on the quantum state and then measures the particular qubits in the measurement positions with the help of the measurement basis and records the measured value as $C'$.
 \item After that the input state is restored by inverting the gates in the reverse order.
 \item Then the user compares whether $$C = C'$$ and accepts if the equality holds, and rejects otherwise.
\end{itemize}
             
\section{Correctness and Complexity Analysis}
\label{Secprob}
In this section, first we are going to argue about the correctness of our algorithm. Then we will analyze the space and time complexities of our algorithm and we shall show that it is possible to implement our algorithm in reasonable amount of resources.

\subsection{Proof of correctness}
Firstly, note that the chosen one-way function for an input $GCH$ state is well defined. As the operations $U_l$'s are unitary, for a definite input $GCH$ state $\Ket{\psi}$ and the chosen one-way function $\mathcal{F}$, the quantum state $C_l$ in the measurement position is unique for every $l$ and hence the classical image $C$, under the chosen one-way function is unique. Thus the one-way function $\mathcal{F}$ is well defined. Also, by the same logic, it is clear that the original input state $\Ket{\psi}$ will always pass the OWF image verification algorithm.
\paragraph*{}In the OWF evaluation algorithm, at each step, there are choices that we have to make under some specified conditions. We will now show that we shall never run out of choices while executing the algorithm. Note that in each of the \textbf{OWF unitary operations} $U_l$'s, the ${(k-1)}^{th}$ and the $k^{th}$ gates are also parallel actions of CNOT operations on compatible pair of qubits.  Hence in order to show that the output state after each \textbf{OWF unitary operation} is a $GCH$ state, it suffices to show that the preconditions for compatible pairs of qubits ensure that after each such CNOT operation, the output state is a $GCH$ state.
\begin{lemma}\label{lemma_1}
If $\Ket{\phi}$ is an $n$-qubit $GCH$ state and CNOT gate is applied on a pair of compatible qubits (as defined in OWF evaluation), then the output quantum state is also an $n$-qubit $GCH$ state.
\end{lemma} 
\begin{proof}
Let the $k^{th}$ and the $l^{th}$ qubits be any arbitrary pair of compatible qubits of $\Ket{\phi}$ and say we apply the CNOT gates on them with the $l^{th}$ qubit as the control qubit. Let this be represented as a unitary transformation $H$ on $\Ket {\phi}$. We need to show that $$H\Ket{\phi} \in \left\{\Ket {\psi^{(n)}}_{GCH}\right\}.$$
Now since the qubits are a pair of compatible qubits, one can verify that exactly one of the following cases occur.
\begin{enumerate}
\item If both the $l^{th}$ and the $k^{th}$ qubits are in $\mathcal{C}$ or $\mathcal{H}$, then the output quantum states of both the $l^{th}$ and the $k^{th}$ qubits will be in $\mathcal{C}$ or $\mathcal{H}$ respectively.
\item If both the $l^{th}$ and the $k^{th}$ qubits are entangled into the same GHZ state which itself is in $\mathcal{B}_m$ $(1<m<n)$, then in the output state, the $l^{th}$ qubit is in the same GHZ state (for $m = 2$, the $l^{th}$ qubit gets transformed into an element of $\mathcal{H}$), but the size of the GHZ state is decreased by one as the $k^{th}$ qubit is disentangled into an element in $\mathcal{C}$.
\item If the $l^{th}$ qubit is in $\mathcal{H}$ and the $k^{th}$ qubit is in $\mathcal{C}$, then the combined output state $\Ket{\phi}_{l,k} \in \mathcal{B}_2$.
\item If the $l^{th}$ qubit is in $\mathcal{C}$ and the $k^{th}$ qubit is in $\mathcal{H}$ or entangled in a GHZ state ($\mathcal{G}$), then both the qubits will remain in the same basis in the output state, as they were before.
\item If the $l^{th}$ qubit is entangled in a GHZ state and the $k^{th}$ qubit is in $\mathcal{C}$, then the $k^{th}$ qubit will get entangled into the same GHZ state in which the $l^{th}$ qubit is entangled, increasing the size of the GHZ state by 1.
\item If the $l^{th}$ qubit is entangled in a GHZ state (of $\mathcal{G}$) and the $k^{th}$ qubit is in $\mathcal{H}$, then both the qubits will remain in the same basis in the output state, as they were before.
\end{enumerate}

Hence we notice that the CNOT operations that we allow on the compatible pair of qubits of a $GCH$ state, keep the total quantum state in $GCH$ state only. Since we started with an $n$-qubit $GCH$ state $\Ket{\psi}$, each of the output states $\Ket{{\phi}_i}$ of the gates $X_i$ $(i=1,2\ldots k)$ is a $GCH$ state. In fact, after the $k^{th}$ gate in each operation $U_{l}$, the qubits in the measurement positions are either elements in $\mathcal{C}$ or $\mathcal{H}$, since the marked qubits are such that they correspond to the Hadamard and computational basis elements used as control and target qubits respectively in the (1) and (2) case as given above, and the combined quantum state of the remaining qubits form a $\frac{n}{2} + 1$ qubit $GCH$ state.

Also, by the previous chart, we conclude that after a CNOT operation on compatible pairs of qubits, the output pair of qubits is also compatible for CNOT operation. Therefore, inverting the CNOT operation on the output state ($U\Ket{\phi}$) of a \textbf{OWF unitary operation} ($U$), in reverse directions, is also a \textbf{OWF unitary operation}. Hence the inverted \textbf{OWF unitary operation} ($U^\dagger$) on the output state ($U\Ket{\phi}$) is also a \textbf{OWF unitary operation}.
\end{proof}

\paragraph*{}The obvious next question is whether or not for each \textbf{OWF unitary operation} $U_l$, we can chose the first $(k-2)$ \textbf{OWF gate operations} obeying the different conditions as specified in the OWF evaluation algorithm. For that, we will show that unless all qubits have been saturated, we can choose compatible pairs of qubits for CNOT operations, obeying all essential properties.

\begin{lemma}\label{lemma_2}
In the OWF evaluation algorithm, while choosing the \textbf{OWF gate operations} for any \textbf{OWF unitary operation} $U_l$, at any point of time if a qubit is unsaturated, then it can be saturated by selecting a suitable \textbf{OWF gate operation} in the next step such that all possible pairs of positions remain \textbf{CNOT-once} after that step.
\end{lemma} 
\begin{proof} By the previous lemma, the output state, after any of the gates $X_i$ $(i=1,2,\ldots,k)$, is a $GCH$ state. Hence, the unsaturated qubit say the $j^{th}$ qubit, can be either an element in $\mathcal{C}$,$\mathcal{H}$ or is entangled in a GHZ state at that time.
\begin{enumerate}
\item If the $j^{th}$ qubit is entangled in a GHZ state, then there shall be at least one more qubit say, the ${l}^{th}$ qubit, which is entangled in the same GHZ state at that time. Then we can apply a CNOT gate operation on the $l^{th}$ and the $j^{th}$ qubit with anyone of them used as the control qubit.

 \item If the $j^{th}$ qubit is an element in $\mathcal{C}$, then we can choose any other qubit, say the $l^{th}$ qubit and apply a CNOT operation on the $l^{th}$ and the $j^{th}$ qubit with any one of them being used as the control qubit.

 \item If the $j^{th}$ qubit is an element in $\mathcal{H}$, then we can choose any other qubit, say the $l^{th}$ qubit and apply a CNOT operation with the $l^{th}$ qubit as the control qubit and the $j^{th}$ qubit as the target qubit.
\end{enumerate}
 Note that in each of the cases, the $l^{th}$ and the $j^{th}$ qubit together form a compatible pair, and after this \textbf{OWF gate operation}, the pair remains \textbf{CNOT-once} (this pair could not have been used previously, since the $j^{th}$ qubit was unsaturated). Hence, all possible pairs of qubits remain \textbf{CNOT-once} after this step. 
\end{proof}

\begin{theorem}\label{thm_1} Unless, all of the qubits have been saturated, we can always find at least one compatible pair of qubits for CNOT operations in the next \textbf{OWF gate operation} such that the pair after this operation remains \textbf{CNOT-once}.
\end{theorem}
\begin{proof}
If all of the qubits have not been saturated, then $$ \exists k \in [1,n] $$ such that $k^{th}$ qubit is not saturated.

Hence by the previous lemma, we can find another qubit say the $l^{th}$ qubit, such that they form a compatible pair for CNOT gate operations in the next \textbf{OWF gate operation}.
\end{proof}

\subsection{Complexity of the algorithm}
In the OWF evaluation algorithm, for each \textbf{OWF unitary operation}, we selected $k$ many \textbf{OWF gate operations}, each consisting of parallel actions of CNOT gates
and we had to keep track of the pairs of qubits that have been used in previous CNOT operations and which qubits have not been used at all so far. The latter can be done by keeping a list of used pairs of qubits and an one dimensional array of length $n$ to keep track whether the qubit has been used before or not and this is iterated $n$ times for each of the \textbf{OWF unitary operations}. The space required for maintaining the array and list is still polynomial.  The number of times we have to update or view the array and the list is exactly equal to the total number of CNOT operations used in the one-way function. Therefore, we can express the time complexity of the OWF evaluation algorithm in terms of the number of CNOT gates used.
\begin{theorem}
In the OWF evaluation algorithm, the number of CNOT gates used is $\mathcal{O}(n^3)$.
\end{theorem}
\begin{proof}
Let $Y$ be the total number of CNOT gates used in the OWF evaluation algorithm. We want to show that $$  Y  \in \mathcal{O}(n^3). $$

Fix $m \in [1,2\dots ,n]$.

For the \textbf{OWF unitary operation} $U_m$,
let $j$ be the number of the CNOT operations used in the first $(k-2)$ gates of the OWF evaluation algorithm and let $l$ be the total number of CNOT operations in the $k$ gates. In the last $(k-1)^{th}$ gate, there can be a maximum of $\frac{n}{2}$ parallel CNOT gates as there are only $n$-qubits, and we saw that there are $\frac{n}{2}$ CNOT operations in the $k^{th}$ gate. Hence,
\begin{equation}\label{eqn_1}
 l \leq j + \frac{n}{2} + \frac{n}{2} = j + n. 
\end{equation}

Each of the CNOT operations in the first $(k-2)$ gates, corresponds to a pair of compatible qubits and by the condition 1.(a) in the OWF evaluation algorithm (that pairs of qubits are used for at most one CNOT operation), we conclude that 
\begin{eqnarray*}
j &\leq& \text{ total number of pairs possible }\\
&=& \frac{n(n-1)}{2}.
\end{eqnarray*}
By \eqref{eqn_1}, we get,
\begin{equation*}  l  \leq  \frac{n(n-1)}{2} + n.  \end{equation*}

This is true for each of the $n$ \textbf{OWF unitary operation} $U_m$.
Therefore, 
$$Y \leq n\left(\frac{n(n-1)}{2} + n\right) = \frac{n^3 + n^2}{2}  \in \mathcal{O}(n^3). $$                       
\end{proof}

Assuming that a CNOT gate takes constant time, we conclude that our OWF evaluation algorithm works in polynomial time.

The OWF image verification algorithm involves use of those CNOT gates, that have been chosen during the OWF evaluation algorithm followed by measurement and comparisons which take only linear time. Hence we can also implement the OWF image verification algorithm in polynomial time.

\paragraph*{} Now, we will analyse the space complexity of the two algorithms. The space required for applying the CNOT gates is proportional to the number of CNOT gates used, hence by previous theorem, only polynomial space is required for executing the gate operations. Next, we will see the amount of space required for storing the description of the one-way function and its image is also polynomial.

\begin{theorem}
The space required for storing the description of the one-way function is $\mathcal{O}(n^3) $ whereas to store the image is $\mathcal{O}(n^2)$.
\end{theorem}
\begin{proof}
 Fix $l \in [1,2\dots ,n]$, arbitrary,
 For the \textbf{OWF unitary operations} $U_l$.
 By the previous theorem,
 $$ \text{the number of CNOT gates used is in } \mathcal{O}(n^2).$$

 As far as the measurement positions and measurement basis are concerned, note that,
  we have to store the POVM and positions for $\frac{n}{2} - 1$ qubits.
Hence,$$ \text{space required for } F_l \in \mathcal{O}(n^2) \text{ , }\forall l = 1,2, \ldots ,n.$$
For $C_l$, we have to store the quantum state of $\frac{n}{2} - 1$ qubits (which can be $\Ket{0},\Ket{1},\Ket{+}$ or $\Ket{-}$).

Hence,$$ \text{space required for } C_l \in \mathcal{O}(n) \text{ ,  }\forall l = 1,2, \ldots ,n.$$
Therefore, $$\text{total  space required for }H \in \mathcal{O}(n^3).$$
and $$\text{total  space required for }C \in \mathcal{O}(n^2).$$  
\end{proof}

 Since both the algorithms demand only the storage of the one-way function description and its image, we conclude that both the algorithms run in $\mathcal{O}(n^3)$ space.

\paragraph*{}Due to low space and time complexity, our proposed one-way function can be feasibly used in various other protocols of quantum crypto currencies for authenticating quantum states.

\section{Security Analysis} Now we shall discuss about the security aspects of our protocol, where the main aim is to ensure the one-wayness of the the proposed one-way function and prove its security against random guesses.

Firstly, notice that, in the OWF Image verification algorithm, if a different quantum state, say $\Ket{\phi} \neq \Ket{\psi}$ (original input $GCH$ state), is used as input and for any \textbf{OWF unitary operation} $U_l$, if the qubits in the measurement positions of $U_l\Ket{\phi}$ is not in the same basis ($\mathcal{B},\mathcal{C}$, or $\mathcal{H}$) as that of the qubits in the measurement of $U_l\Ket{\psi}$, then measuring through the POVM as recorded in the measurement basis would result in the destruction of the state $U_l\Ket{\phi}$. Hence, the primary goal for an adversary is to produce a state such that the qubits in the measurement position of this state are in the same basis as that of the qubits in the measurement positions of $U_l\Ket{\psi}$.

Also, note that in our OWF evaluation algorithm, for any pair of positions, say $(k,l)$, we have used at most two CNOT operations on the qubits at these positions. Moreover, if there are two CNOT operations on the $k^{th}$ and the $l^{th}$ qubits, the second CNOT operation is due to a \textbf{CF} operation on that pair of positions and hence the two CNOT operations do not nullify each other. This is because CNOT gates are self invertible (i.e., $U^2 = UU^{\dagger} = I$) and non-symmetric about the qubits.

\subsection{Non-invertibility of the one-way function}
Now, we will show that the way we defined the operations in the OWF evaluation algorithm, it is not possible to invert the gates and get back the input state.
Finding the input state $\Ket {\psi}$ by inverting the gates basically means constructing the output state $U_l\Ket{\phi}$ completely or partially and then applying the inverted gates on it to get the input state completely or partially, i.e., finding the quantum state of qubits at some positions. Since half the output state is revealed, we need to ensure that by the knowledge of half the qubits and applying the inverted gates, one should not be able to get back the input state, not even partially. Although both the OWF evaluation and verification algorithm involves measurement, it is not straightforward why inversion of gates is not possible. This is because the output state $U_l\Ket{\phi}$ is preserved after the measurement and hence measurement does not really play a role in terms of non-invertibility of gates.
\begin{theorem}
For any \textbf{OWF unitary operation} $U_l$ ($\forall 1\leq l\leq n$) as defined in the OWF evaluation algorithm, given the value of $F_l$ and $C_l$ it is not possible for an adversary to find the quantum state of the qubits of the input $GCH$ state $\Ket {\psi}$ in polynomial time.
\end{theorem}
\begin{proof}
 Fix $l$, arbitrary such that $1 \leq l \leq n$.
Let $U_l, C_l, X_j, \Ket{\phi}_j$ has the same meaning as before.

For the operation $U_l$, let us analyse the step in the OWF evaluation algorithm, where we choose $\frac{n}{2}$ CNOT operations for the last gate $X_k$, and then mark the qubits from which we chose $(\frac{n}{2} - 1)$ marked qubits to be revealed as $C_l$ in $C$.

The pair of qubits on which the CNOT operations took place in the $k^{th}$ operation, such that at least one of the output qubit were marked, are either of the following forms:
\begin{enumerate}
\item Both were entangled in the same GHZ state, which itself is an element of $\mathcal{B}_m$.
\item Both are in $\mathcal{C}$.
\item Both are in $\mathcal{H}$.
\end{enumerate}
 In the first case, if $m=2$, we mark the control qubit in the output state, else we mark the target qubit. Note that, the target qubit is always transformed into an element in $\mathcal{C}$ after the CNOT operation, but the control qubit is transformed into an element in $\mathcal{H}$ if $m=2$, otherwise remains entangled in a GHZ state belonging to $\mathcal{B}_{m-1}$. Hence the marked qubit in the output state after the $k^{th}$ operation is either an element in $\mathcal{H}$ (if  $m=2$) which was used as the control qubit or an element in $\mathcal{C}$ (for $m > 2$) which was used as the target qubit in the CNOT operations of $X_k$.   

Similarly in the second and third cases, the marked qubits were used as target qubits and control qubits respectively in the CNOT operations. After the CNOT operation, the marked qubits remain to be an element in $\mathcal{C}$ and $\mathcal{H}$ in the second and third cases respectively.

 We know that, in $X_k$ there are $\frac{n}{2}$ CNOT operations and for all but one CNOT operation, there exists a single marked qubit, quantum state of which in the output state is revealed through $C_l$. For the remaining one CNOT operation, none of the output qubits is revealed, hence there can be many possible states of the input qubits, clearly more than 2. 

\paragraph*{}
Hence with the knowledge of the quantum state of the marked qubits (in $C_l$) each of which can either be an element in $\mathcal{C}$ or $\mathcal{H}$, there are at least two possible bases for the input pair of qubits of each of the CNOT operations.  Therefore, by the knowledge of $C_l$,
    $$\# ( \text{quantum states for } \Ket{\phi}_{(k-1)} ) \geq 2^{\frac{n}{2}},$$
each of which are equally likely. Hence there are at least $2^{\frac{n}{2}}$ possible quantum states for $\Ket{\phi}_{(k-1)}$ and no two states out of these $2^{\frac{n}{2}}$ possibilities, have the same POVM for all of their qubits. Since we have constructed the OWF evaluation algorithm in such a way that in each \textbf{OWF unitary operation}, the CNOT gate operations do not nullify each other, one must figure out all the intermediate $\Ket {\phi}_{i}$ $\forall i = 1,2,\ldots,(k-1)$ in order to get $\Ket {\psi}$.

\paragraph*{} Hence, in order to recover the possible quantum state of any qubit of the original input state $\Ket{\psi}$ or even the bases of its qubits, an adversary must scan through all the $2^{\frac{n}{2}}$ possibilities, each of which is equally likely, and hence shall take exponential time. 
\end{proof}

If the adversary randomly guesses the unknown qubits of the output state for the \textbf{OWF unitary operation} $U_l$, then $$ \Pr \left(\text{adversary gets back the correct input state} \Ket{\phi} \right)  \leq \frac{1}{2^{\frac{n}{2}}}$$ which is negligible.

\subsection{Security against random guess}

\paragraph*{} 
\paragraph*{}Another kind of attack that is possible in our one-way function is by random guess of unknown qubits. Note that for each of the $U_l$, by the knowledge of $C_l$ and $F_l$, an adversary can only know about $\frac{n}{2} - 1$ qubits of the output state. Hence, for a particular \textbf{OWF unitary operation} $U_j$, it can guess the quantum state of the unknown $\frac{n}{2} + 1$ qubits to be any $\frac{n}{2} + 1$ qubit state and then invert the gates of $F_j$ to get a state, say $\Ket{\phi}_j$, which might not be same as the original input state $\Ket{\psi}$ but satisfies the equation, 
\begin{equation}\label{eqn_2} 
F_j\Ket{x} = F_j\Ket{\psi}.
\end{equation}
 But in order to have $$\mathcal{F}\Ket{\phi}_l = \mathcal{F}\Ket{\psi}$$ it must satisfy the equation,
\begin{equation} \label{eqn_3}
F_l\Ket{x} = F_l\Ket{\psi}, \quad \forall l =1,2,\ldots,n. 
\end{equation}

But note that, there are infinitely many quantum states possible for the combined state of $n$-qubits. Hence, it does not make sense for an adversary to randomly select from an infinite space. So the best an adversary can do is to select from a finite set of quantum state for a random guess of the unknown qubits of $U_l$. We shall show that the probability of success of such an attack is exponentially low, but before that let's prove some useful lemmas.

 Let $P_j$ be the probability that any randomly chosen $n$-qubit $GCH$ state will satisfy Equation \eqref{eqn_2} and $p_j$ be the probability that any randomly chosen quantum non-$GCH$ state will satisfy Equation \eqref{eqn_2}.
\begin{lemma}\label{lemma_3}
For any \textbf{OWF unitary operation} $U_j$, $p_j = 0$.
\end{lemma}
\begin{proof}
 Let $S$ be the set of all $n$-qubit quantum states and let $B$ denotes the set of all $n$-qubit quantum states satisfying Equation \eqref{eqn_2} and let $$M = \left\{\Ket {\psi^{(n)}}_{GCH}\right\}^{\mathsf{c}}.$$
 Thus,
 $$ p_j =  \Pr\left( \Ket{\xi} \in B ~\left.\right|~ \Ket{\xi} \xleftarrow{\$} M\right).$$ Here, $\Ket{\xi} \xleftarrow{\$} M$ denotes that $\Ket{\xi}$ is selected from $M$ at random.
Note that, $$\left.{U_j}\right|_B : B \rightarrow  S$$ as a unitary transformation from a subspace of the $2^n$ dimensional Hilbert space to the Hilbert space.
Also by definition of $B$,
    $\forall \Ket{\xi} \in B$, $U_j\Ket{\xi}$ has the property that it is the tensor of a $2^{(\frac{n}{2} + 1)}$ tuple vector and a $2^{(\frac{n}{2} - 1)}$ tuple vector out of which the later is fixed and the former can be any $2^{(\frac{n}{2} + 1)}$ tuple vector. Hence,
    $$Rank(\left.{U_j}\right|_B) = 2^{(\frac{n}{2} + 1)} < 2^n = dim(S).  \text{\space(assuming $n>2$)}$$
Since $U_j$ is unitary and hence invertible, by {\em rank-nullity} theorem. 
$$ dim(B) =  Rank(\left.{U_j}\right|_B) < dim(S) .$$

Since, the total $2^n$-dimensional Hilbert space is a finite dimensional vector space over the infinite field $\mathbb{C}$.
$$\Pr\left( \Ket{\xi} \in B ~\left.\right|~ \Ket{\xi} \xleftarrow{\$} S\right) = 0.$$ 
In order to see this, consider any finite dimensional vector space $V$ of dim $n$, over a finite field $F$ of size $p = |F|$. Let $W$ be a subspace of $V$ of dimension $m$ ($< n$).
Hence, $|V| = {p}^{n}$ and $|W| = {p}^{m}$.
Let $$P = \Pr\left( x \in W ~\left.\right|~ x \xleftarrow{\$} V\right).$$ Clearly, $$P = \frac{|W|}{|V|} = \frac{1}{p^{(n-m)}}.$$
In our case, we are concerned with a finite dimensional vector space over an infinite field. Hence it suffices to take the limit $p \to \infty$.
Now, $$\lim_{p \to \infty} P =\lim_{p \to \infty} \frac{1}{p^{(n-m)}} = 0 $$ as $n>m$. Hence, $\Pr\left( \Ket{\xi} \in B ~\left.\right|~ \Ket{\xi} \xleftarrow{\$} S\right) = 0.$

In our case, the set $\left\{\Ket {\psi^{(n)}}_{GCH}\right\}$ has finite cardinality whereas the total set $ S$ is an infinite set. Hence, 
$$ \Pr\left( \Ket{\xi} \in B ~\left.\right|~ \Ket{\xi} \xleftarrow{\$} M\right) =  \Pr\left( \Ket{\xi} \in B ~\left.\right|~ \Ket{\xi} \xleftarrow{\$} S\right) = 0.$$ 
\end{proof}

\begin{lemma}\label{lemma_4}
For any \textbf{OWF unitary operation} $U_l$, there can be  $2^{(\frac{n}{2} - 1)}$ possible values of $C_l$ for the same $F_l$ and for each such possible values of $C_l$'s, $\exists \Ket{\phi} \in \left\{\Ket {\psi^{(n)}}_{GCH}\right\}$ such that $F_l\Ket{\phi} = C_l$.
\end{lemma}
\begin{proof}
Since $F_l$ is fixed, the measurement positions and measurement bases are fixed. The quantum state of some of the qubits in the output state is in $\mathcal{C}$ and some in $\mathcal{H}$. Keeping the basis same, a particular qubit in the measurement position can either be $\Ket{0}$ or $\Ket{1}$ (if the qubit is supposed to be an element in $\mathcal{C}$); or can be $\Ket{+}$ or $\Ket{-}$ (if the qubit is supposed to be an element in $\mathcal{H}$). 

Since there are ${\frac{n}{2} - 1} $ qubits in the measurement positions, there can be $2^{\frac{n}{2} - 1} $ possible combined states of qubits in measurement positions.

Since, $C_l$ is uniquely defined by the quantum states of the qubits in the measurement positions of $U_l$ and vice versa,
\begin{equation}\label{eqn_4} \# (\text{ values of $C_l$}) = 2^{\frac{n}{2} - 1} = s (\text{say}).\end{equation}

Let the quantum states corresponding to all such possible values of $C_l$ be listed as $\Ket{\alpha^1},\Ket{\alpha^2},\ldots,\Ket{\alpha^s}.$ Note that
$$\exists t\text{ such that } {\alpha}^{t} = F_l\Ket{\psi} \text{ where $\Ket{\psi}$ is the original input.}$$

Let the quantum state of the other ${\frac{n}{2} + 1} $ qubits of $U_l\Ket{\psi}$ be represented by ($\Ket{\beta}$).

$ \text{For } m = 1,2,\ldots,s$, construct a quantum state $\Ket{\zeta^m}$ keeping $\Ket{\alpha^m}$ at the measurement position and $\Ket{\beta}$ for the other ${\frac{n}{2} + 1} $ positions.
Hence, up to some rearrangement,
$ \Ket{\zeta^m} = (\Ket{\alpha^m}\otimes\Ket{\beta})$.
Let $$ \Ket{\xi^m} = {U_l}^\dagger \Ket{\zeta^m} .$$
Hence,
\begin{equation}\label{eqn_5} \Ket{\xi^t} = {U_l}^\dagger (\Ket{\alpha^t}\otimes\Ket{\beta}) = \Ket{\psi}\end{equation}
and
$$ U_l\Ket{\xi^m} = \Ket{\zeta^m} = (\Ket{\alpha^m}\otimes\Ket{\beta}) \text{  (up to some rearrangement)}.$$
\begin{equation}\label{eqn_6} \implies F_l\Ket{\xi^m} = {\alpha}^{m}\quad \forall m=1,2,\ldots,s.\end{equation}
Following the proof of Lemma \ref{lemma_1}, we note that, after a \textbf{OWF gate operation} on a $GCH$ state, the basis of the output qubits depends only on the basis of the input qubits. 
Since the measurement basis is fixed and the inverted \textbf{OWF unitary operation} ${U_l}^\dagger$ on $\Ket{\zeta^m}$ ($m = 1,2\ldots s$),  is also a \textbf{OWF unitary operation} we conclude that if  
$$\exists {\alpha}^m \text{ such that } \Ket{\xi^m} \in \left\{\Ket {\psi^{(n)}}_{GCH}\right\} ,$$ then
$$ \Ket{\xi^m}
\in \left\{\Ket {\psi^{(n)}}_{GCH}\right\}    \forall m=1,2,\ldots ,s. $$

But by Eqn \eqref{eqn_5}, $$\Ket{\xi}^t = \Ket{\psi}.$$

Therefore,
\begin{equation}\label{eqn_7}
\Ket{\xi^m} \in \left\{\Ket {\psi^{(n)}}_{GCH}\right\}    \forall m=1,2,\ldots ,s. 
\end{equation}
By equations \eqref{eqn_4}, \eqref{eqn_6} and \eqref{eqn_7}, we are done. 
\end{proof}

\begin{corollary}
For a random $n$-qubit $GCH$ state and a fixed \textbf{OWF unitary operation} $U_l$, there can be at least $2^{\frac{n}{2} - 1}$ distinct possible values of $C_l$.
\end{corollary}
\begin{proof}
By the above lemma, keeping $F_l$ fixed, there are $2^{\frac{n}{2} - 1}$ $n$-qubit quantum states which produce $2^{\frac{n}{2} - 1}$ distinct values of $C_l$. Hence, for any random $n$-qubit $GCH$ state, there can be at least $2^{\frac{n}{2} - 1}$ distinct possible values of $C_l$.
\end{proof}
\begin{lemma}\label{lemma_5}
For each of the \textbf{OWF unitary operations} $U_j$, $P_j \leq \frac{1}{2^{\frac{n}{2} - 1}}.$ 
\end{lemma}
\begin{proof}
 By the previous corollary, for any random $n$-qubit quantum state,
  $$ \# (\text{quantum states of the qubits in the measurement position of } U_j) \geq 2^{\frac{n}{2} - 1}.$$ In order to satisfy Equation \eqref{eqn_2}, the value in the measurement position must be exactly same as $C_j$. Hence, 
 $ P_j \leq \frac{1}{2^{\frac{n}{2} - 1}}. $ 
\end{proof}

\begin{theorem}
For each \textbf{OWF unitary operation} $U_m$, the probability that the state thus constructed by an adversary by randomly guessing the unknown qubits of $U_m$, satisfies Equation \eqref{eqn_3}, is negligible.  
\end{theorem}
\begin{proof}
Let $P^m$ be the probability that the state $\Ket{\phi}_m$ thus constructed by randomly guessing the unknown qubits of $U_m$, satisfies Equation \eqref{eqn_3}. Let $X_m$ be the finite set from which the adversary randomly guesses the unknown qubits and $Y_m$ be the finite set of input states which can be constructed by randomly guessing from $X_m$. Let $\Ket{\phi}_m$ be an arbitrary element of $Y_m$. Let $Z = \left\{\Ket {\psi^{(n)}}_{GCH}\right\}.$

Note that,
\begin{equation*}
\begin{split}
 & \Pr\Bigl(\Ket{\phi}_m \text{ satisfies Eqn. }\eqref{eqn_2} ~\left.\right|~ \Ket{\phi}_m \in Z \Bigr)\\
 & = \Pr\left(\Ket{\phi} \text{ satisfies Eqn. }\eqref{eqn_2} ~\left.\right|~ \Ket{\phi} \xleftarrow{\$} Z, \Ket{\phi} \in Y_m\right).
\end{split}
\end{equation*}

Since the $n$ \textbf{OWF unitary operations} are chosen independently, for $m \neq j$ (the index of the unitary in Equation~\eqref{eqn_2}),
the above probability is equal to
\begin{equation*}
 \Pr\left(\Ket{\phi} \text{ satisfies Eqn. }\eqref{eqn_2} ~\left.\right|~ \Ket{\phi} \xleftarrow{\$} Z \right) = P_j.
\end{equation*}
Similarly,
$$\Pr\Bigl(\Ket{\phi}_m \text{ satisfies Eqn. }\eqref{eqn_2} ~\left.\right|~ \Ket{\phi}_m \in Z^{\mathsf{c}}\Bigr) = p_j.$$
Obviously, for $j = m$,
$$\Pr\Bigl(\Ket{\phi}_m \text{ satisfies Eqn. }\eqref{eqn_2} ~\left.\right|~ \Ket{\phi}_m \in Z\Bigr) = 1,$$
$$\Pr\Bigl(\Ket{\phi}_m \text{ satisfies Eqn. }\eqref{eqn_2} ~\left.\right|~ \Ket{\phi}_m \in Z^{\mathsf{c}}\Bigr) = 1.$$
Again, since the $n$ \textbf{OWF unitary operations} are chosen independently, and by previous discussion,
$$\Pr\Bigl(\Ket{\phi}_m \text{ satisfies Eqn. }\eqref{eqn_3} ~\left.\right|~ \Ket{\phi}_m \in Z\Bigr) = 1.\left({\prod}_{j\neq m} P_j\right). $$
$$\Pr\Bigl(\Ket{\phi}_m \text{ satisfies Eqn. }\eqref{eqn_3} ~\left.\right|~ \Ket{\phi}_m \in Z^{\mathsf{c}}\Bigr) = 1.\left({\prod}_{j\neq m} p_j\right) . $$
Therefore, by Lemma \ref{lemma_3} and Lemma \ref{lemma_5}, we conclude,
$$\Pr\Bigl(\Ket{\phi}_m \text{ satisfies Eqn. }\eqref{eqn_3} ~\left.\right|~ \Ket{\phi}_m \in Z\Bigr) \leq \left(\frac{1}{2^{\frac{n}{2} - 1}}\right)^{(n-1)}. $$

$$\Pr\Bigl(\Ket{\phi}_m \text{ satisfies Eqn. }\eqref{eqn_3} ~\left.\right|~ \Ket{\phi}_m \in Z\Bigr) = 0.$$
Let $$ P^{(m)} = \Pr\Bigl(\Ket{\phi}_m \text{ satisfies Eqn. }\eqref{eqn_3} ~\left.\right|~ \Ket{\phi}_m \in Z\Bigr),$$
  $$ p^{(m)} = \Pr\Bigl(\Ket{\phi}_m \text{ satisfies Eqn. }\eqref{eqn_3} ~\left.\right|~ \Ket{\phi}_m \in Z^{\mathsf{c}}\Bigr).$$
Therefore,
\begin{eqnarray*}
 P^m &=& P^{(m)}\Pr\Bigl(\Ket{\phi}_m \in Z\Bigr) + p^{(m)}\Pr\Bigl(\Ket{\phi}_m \in Z^{\mathsf{c}}\Bigr) \\
     &=& P^{(m)}\Pr\Bigl(\Ket{\phi}_m \in Z\Bigr)     \leq P^{(m)} \leq  \left(\frac{1}{2^{\frac{n}{2} - 1}}\right)^{(n-1)}.\\
\end{eqnarray*}

This concludes our proof.
\end{proof}

Hence,$$ \text{the success probability of forgery by this method} \leq  \left(\frac{1}{2^{\frac{n}{2} - 1}}\right)^{(n-1)},$$ which is negligible.
\section{Application of the proposed OWF  in Quantum Money Schemes}
In the post quantum era, there have been some great works in the field of quantum money. Wiesner's private key quantum scheme was the first of its kind which was proven to be information theoretically secured~\cite{wie78}. Later, Farhi also came up with the quantum money scheme based on knot theory~\cite{farhi}. Also, Aaronson and Christiano proposed the idea of hidden subspace for constructing public key quantum money schemes~\cite{hidsub}. Recently, Panigrahi and Moullick also came up with the idea of quantum cheques using entanglement and quantum teleportation~\cite{qcheq}.

 The classical one-way function had a major contribution in developing the data authentication models in cryptography. Similarly, the concept of quantum to classical one-way function  would enhance the facility of authenticating a quantum state and with the help of secured digital signature schemes, would provide the ability to authenticate any quantum state publicly multiple times. This property is the main reason why the Quantum to Classical OWF  is so useful in the field of quantum money. Our proposed one-way function  is multiple times verifiable, as the verification does not destroy the original quantum state.

  Hence it is of great use in public key quantum money schemes like quantum bitcoins and quantum currencies which demand multiple times verification without destroying the original state. However, it cannot be used in quantum cheques or currency bonds since the multiple times verifiability property can be exploited for double-spending in such schemes. Ideally, any currency bond scheme should be only one time verifiable, in order to avoid double spending.

   Now, we will give explicit constructions of quantum currency and quantum bitcoins scheme based on the proposed one-way function. The security of both the schemes can be explained on the basis of the no-cloning principle~\cite{noclon}, the security of the quantum to classical one-way function  and the security of the digital signatures. 

\subsection{Quantum Currency Notes} In this section, we are going to discuss a centralised yet public-key quantum money scheme called Quantum Currency Notes where a central authority mints the currency notes but these notes are publicly verifiable multiple times. The explicit protocol is shown below.
\subsubsection{Minting Algorithm}
In the minting session, the central authority or Bank mints the quantum currency notes. The algorithm followed is listed below.
\begin{enumerate}
\item The Bank first performs the key generation algorithm of the digital signature scheme and produces the key pair $(K_{priv}, K_{pub})$. The key $K_{pub}$ is made public. 
\item The Bank creates a random quantum GCH state $\Ket {\psi}$ and runs the OWF evaluation algorithm and selects a suitable one-way function $F$ during the execution and produces the classical image $L$ of the quantum state $\Ket {\psi}$ under the chosen one-way function $F$.
\item Let the worth value of the quantum currency note be represented classically by $\mathcal{M}$. The Bank signs the classical string $I = (\mathcal{M} \| L\|F)$ with its private key $K_{priv}$ to get the signature state $\sigma$. The state $\$ = (\Ket {\psi}, \mathcal{M}, L, F, \sigma)$ is the quantum currency note that the Bank gives to the user. 
\end{enumerate}

\subsubsection{Verification Algorithm}
Before accepting the quantum currency note, the user performs the following algorithm. 
\begin{enumerate}
\item The user after receiving the state verifies that the state is of the form  $
(\Ket {\psi}, \mathcal{M}, L, F, \sigma)$ or not. 
\item Then the user recovers the classical text $I = (\mathcal{M} \| L\|F)$. After that, with the help of $K_{public}$, the user verifies whether the signature state $\sigma$ matches with the message text $I$ or not. If the signature state $\sigma$ does not match with $I$, then the user aborts the verification and declares the quantum currency note to be invalid. 
\item Then the user runs the OWF Verification Algorithm on the state $\Ket {\psi}$ with $F$ as the one-way function description and $L$ as the one-way function image and thus checks whether the three match or not. If they do not pass the OWF Verification Algorithm, the user aborts the verification and the currency note is declared to be invalid. 
\end{enumerate}

If and only if all the three steps are passed successfully, the user accepts the quantum currency note and the verification is said to be complete. 

It can be seen that this scheme is secured on the basis of no-cloning theorem~\cite{noclon} and the security of the digital signatures and the secured Quantum to Classical OWF. The main use of the no-cloning theorem~\cite{noclon} in the entire protocol is to ensure that anybody having the currency note consisting of a quantum state (unknown to him) and some classical data, cannot copy or know the quantum state and hence cannot reproduce the currency note and therefore cannot do forgery.
\subsection{Quantum Bitcoins}
Ever since its very origin, the popularity of bitcoins have gathered a lot of pace across the globe for its decentralised network and anonymity of transactions. Being a public-key quantum money, the quantum bitcoin scheme demands multiple times verifiability. Hence, the quantum to classical one-way function  is an ideal tool for such a scheme. 
                       
    In recent times, some schemes have been devised on quantum bitcoins, that have the advantage that unlike the classical bitcoins the transactions are not required to be recorded~\cite{qbitcoin} and thereby saving on a lot of space and time. In our quantum bitcoin protocol, at the mining session two miners are chosen randomly which can be chosen by the system itself by a lottery or through a competition as it is done in case of the real bitcoins just for minting the bitcoins. The concept of incentives for the miners is still valid for our scheme. The minting and verification algorithms are as follows.
\subsubsection{Minting Algorithm}
Let Alice and Bob be the miners selected for the minting session who mint a bitcoin and stores some of its verification tools in the public classical blockchain ledger. \newline
The steps of the minting algorithm are as follows.
\begin{enumerate}
\item At first, Bob runs the Key generation algorithm of digital signatures to produce its own pair of private-public keys $(K^{*}_{priv}, K^{*}_{pub})$ and sends 
$K^{*}_{pub}$ to Alice. Alice also runs key generation algorithm to produce its own private key and public key $(K_{priv}, K_{pub})$. 
\item Then, Bob prepares two copies of a random $n$-qubit quantum GCH state and sends one of the copies to Alice. Let the state be of the form $\Ket {K}$. 
\item Alice applies a random quantum gate $G$ (which can be realised through polynomially many universal gates) on its copy of $\Ket {K}$ to produce the state
$\Ket {K^{*}} = G\Ket {K}$, a randomised state and tags with a serial number $S_{r}$. Let the description of the gate be given by $G$, a classical information. Then Alice signs on the classical information $S_{r}\|G$ with its $K_{priv}$ to get the signature state $\sigma_{1}$ and keeps the corresponding $K_{pub}$ with herself. She thus forms the token $\$_{1} = (S_{r}, G, \Ket {K^{*}}, \sigma_{1}$). 
\item Bob runs the OWF evaluation algorithm during which it selects a suitable Quantum to Classical one-way function $F$ and computes the classical image $L$ of the quantum state under the chosen one-way Function i.e., $$L = F(\Ket {K}).$$ He signs the classical information $(F\|L)$ with the help of $K^{*}_{priv}$ and produces the signature state $\sigma_{2}$. He thus forms the token $\$_{2} = (\Ket {K}, L, F, \sigma_{2}$) and sends it to Alice. 
\item Alice stores $S_{r}$ and the corresponding $K_{pub}$ and $K^{*}_{pub}$ in her ledger which she does not disclose immediately. 
\item The whole token $\$ = (\$_{1}, \$_{2})$ is the quantum bitcoin. 
After making $m$ such bitcoins Alice uploads the $m$ serial numbers and the respective information into the public classical blockchain ledger and releases the quantum bitcoins to the users each of which verifies it before accepting it. 
\end{enumerate}
		
\subsubsection{Verification Algorithm}
Before accepting a bitcoin a user shall run a verification protocol which can be summed up in the following steps:
\begin{enumerate}
\item After receiving the bitcoin \$ at first the user checks whether the bitcoin is of the valid form or not, i. e., it checks whether $\$ = (\$_{1}, \$_{2})$ or not and whether $\$_{1} = (S_{r}, G, \Ket{K^{*}}, \sigma_{1})$ or not and whether $\$_{2} = (\Ket {K}, L, F, \sigma_{2})$. If not, it aborts the verification. 
\item Then with the help of the $K^{*}_{pub}$, the user verifies whether the signature state $\sigma_2$ and the classical information $(F\|L)$ matches or not. Similarly, with the help of the $K_{pub}$, the user verifies whether the signature state $\sigma_1$ and the classical information $(S_{r}\|G)$ matches or not. If any of the two do not match, the verification is aborted.
\item After that with the help of the classical information $F$ and $L$, it runs the one-way function verification algorithm on the state $\Ket {K}$  with  $F$ as the one-way function description and $L$ as the one-way function image. If the one-way function  verification fails, the user aborts the verification.  
\item Then from the knowledge of the classical information $G$, it applies the quantum gate $G$ on the state $\Ket {K}$ to get the state $\Ket {K'} = G\Ket {K}$ and verifies whether $\Ket {K'} = \Ket {K^{*}}$ by the help of non-destructive comparison of states. If they do not match, the user aborts the verification. 
\end{enumerate}
The user concludes it to be a valid bitcoin and accepts it, if and only if the bitcoin passes all the four verification steps. 

\paragraph*{}Our scheme is secured under reuse attack, mentioned in the previous quantum bitcoins protocol~\cite{qbitcoin}, i.e., the same exact quantum bitcoin cannot be minted by the previous miners in order to cause inflation. Bob cannot know which bitcoin has he minted as the serial numbers are not disclosed to him, whereas Alice does not know the random state Bob had prepared during minting. So as long as they do not work together they cannot launch the re-use attack and since the miners are randomly chosen, it is very less probable that both of them would collaborate to launch the attack.

In any bitcoin scheme, minting is done only once for a definite bitcoin but the verification is done every time after a transaction is done and hence verification should be fast and simple. Our scheme offers a quite fast and simple verification procedure, as both OWF image verification and verification of the digital signatures take polynomial time.

                            As mentioned earlier, the security of our scheme relies on the no-cloning property~\cite{noclon}, the security of the digital signatures and the Quantum to Classical OWF. 
 
\section{Conclusion \& Future Works}
We have put forward the concept of Quantum to Classical OWF that is very useful in authenticating a particular kind of quantum states. Although the scheme seems analogous to the classical signature schemes, it is more secure in some sense due to no-cloning property~\cite{noclon} which ensures that it is not possible for any counterfeiter to copy unknown quantum states. Hence, it is going to be very useful in money schemes as shown above. In fact, almost any classical money scheme can be modelled into a quantum money scheme with the help of a Quantum to Classical OWF. Just like digital signatures help to authenticate classical strings, the Quantum to Classical OWF  allows one to authenticate quantum states. Of course, the scheme in the current form is theoretical and practical implementation of such a scheme brings in many other challenges like decoherence and Quantum error-correction.

Also in our scheme, for each OWF unitary operation we have revealed the classical measured state of half the qubits. Although it might seem that, if we decrease the number of qubits to be revealed, the security of the scheme would be enhanced, but this is not entirely true. It is true that by doing that, it would be harder to recover the qubits of the original input by inversion of gates but at the same time there will be an increase in the probability of success for an adversary to construct a quantum state, capable of passing the OWF verification algorithm. Thus, it is a good question to ask what is the optimal strategy or optimal number of qubits to be revealed in the measurement position of each of the OWF unitary operations.

Another thing to notice in our scheme is that it does not allow authentication for an arbitrary quantum state. This is basically because there is no known way to simulate any given unitary operation through polynomially many quantum universal gates. A natural question to ask is whether it is possible to somehow modify the scheme to include more quantum states? Further study is required to address such interesting but open questions.

\end{document}